\documentclass[a4paper,11pt]{article}

\usepackage{fullpage}

\usepackage{algorithm}
\usepackage{algpseudocode}

\usepackage{booktabs} 
\usepackage{cite}

\usepackage{tikz}

\usepackage{latexsym}
\usepackage{url}
\usepackage{listings}
\usepackage{xspace}
\usepackage{color}

\usepackage{amsmath}
\usepackage{amssymb}
\usepackage{amsthm}

\newcommand{\eg}{e.g.\@\xspace}

\newcommand{\ceil}[1]{\lceil #1\rceil}

\newcommand{\mpisendrecv}{\textsf{MPI\_\-Sendrecv}\xspace}
\newcommand{\mpiisendrecv}{\textsf{MPI\_\-Isendrecv}\xspace}

\newcommand{\mpibcast}{\textsf{MPI\_\-Bcast}\xspace}
\newcommand{\mpigather}{\textsf{MPI\_\-Gather}\xspace}
\newcommand{\mpiscatter}{\textsf{MPI\_\-Scatter}\xspace}

\newcommand{\mpialltoall}{\textsf{MPI\_\-Alltoall}\xspace}

\newcommand{\mpireduce}{\textsf{MPI\_\-Reduce}\xspace}
\newcommand{\mpireducescatterblock}{\textsf{MPI\_\-Reduce\_\-scatter\_\-block}\xspace}
\newcommand{\mpireducescatter}{\textsf{MPI\_\-Reduce\_\-scatter}\xspace}
\newcommand{\mpiallreduce}{\textsf{MPI\_\-Allreduce}\xspace}

\newcommand{\bidirec}[2]{\textsf{Send}(#1)\parallel\textsf{Recv}(#2)\xspace}

\newtheorem{theorem}{Theorem}

\newtheorem{corollary}{Corollary}

\title{Optimal, Non-pipelined Reduce-scatter and Allreduce Algorithms}

\author{Jesper Larsson Tr\"aff\\
  TU Wien\\
  Faculty of Informatics\\
  Institute of Computer Engineering, Research Group Parallel Computing 191-4\\
  Treitlstrasse 3, 5th Floor, 1040 Vienna, Austria}
\date{October 2024, Revised February 2025}

\begin{document}

\maketitle

\begin{abstract}
  The reduce-scatter collective operation in which $p$ processors in a
  network of processors collectively reduce $p$ input vectors into a
  result vector that is partitioned over the processors is important
  both in its own right and as building block for other collective
  operations. We present a surprisingly simple, but non-trivial
  algorithm for solving this problem optimally in $\ceil{\log_2 p}$
  communication rounds with each processor sending, receiving and
  reducing exactly $p-1$ blocks of vector elements. We combine this
  with a similarly simple, well-known allgather algorithm to get a
  volume optimal algorithm for the allreduce collective operation
  where the result vector is replicated on all processors. The
  communication pattern is a simple, $\ceil{\log_2 p}$-regular,
  circulant graph also used elsewhere. The algorithms assume the
  binary reduction operator to be commutative and we discuss this
  assumption. The algorithms can readily be implemented and used for
  the collective operations \mpireducescatterblock, \mpireducescatter
  and \mpiallreduce as specified in the MPI standard. We also observe
  that the reduce-scatter algorithm can be used as a template for
  round-optimal all-to-all communication and the collective \mpialltoall
  operation.
\end{abstract}

\section{Introduction}

Collective \emph{combine} or \emph{reduction operations} in which sets
of processors (processes) cooperate to globally combine or reduce sets
of input vectors and distribute the result in various ways across the
processors are important algorithmic building blocks for applications
on large-scale, parallel computing systems; for recent applications of
mostly known algorithms, see,
\eg,~\cite{CastelloCatalanDolzQuintanaDuato23,NuriyevManumachuAseeriVermaLastovetsky24,NguyenWahibTakano21}.

Given $p$ consecutively ranked processors $r, 0\leq r<p$ with input
vectors $V_r$ with the same number of elements and an associative,
binary operator $\oplus$ that can element wise combine two vectors, the
global combine or reduction problem is to compute
\begin{eqnarray*}
  W & = & \bigoplus_{r=0}^{p-1}V_r \quad .
\end{eqnarray*}
By the associativity of $\oplus$, brackets can be left out.  If the
operator is also commutative, the order of the input vectors does not
matter. The result vector $W$ can be stored at either a designated
root processor, at all processors, or be partitioned into $p$ blocks
of elements with block $W[r]$ of the result stored at processor
$r$. Blocks may have the same or different numbers of vector elements.
These problems are solved by the MPI collectives \mpireduce (reduction
to a designated root process), \mpiallreduce (result vector replicated
on all processes) , \mpireducescatterblock (result vector
partitioned into blocks having exactly the same number
of elements) and \mpireducescatter (result vector partitioned into
blocks of possibly different sizes)~\cite{MPI-4.1}.  The commonly
used term \emph{reduce-scatter} is somewhat unfortunate, since it
suggests that the problem is solved in two stages as a reduction to a
root processor followed by a scatter operation for partitioning the
result vector. Good algorithms do not take this detour, but solve the
problem directly.  A better intuition is therefore to view the
operation as $p$ simultaneous, rooted reductions with each processor
$r$ being the root in a reduction of the blocks with index $r$.  The
reduce-scatter algorithm to be presented in the following follows this
intuition, but the processors cooperate subtly in the reduction of the
blocks. We will alternatively refer to the reduce-scatter operation as
\emph{partitioned all-reduce}.
  
Our algorithms work uniformly for any number of processors $p$. We
assume the operator $\oplus$ to be commutative. The processors
communicate in a $\ceil{\log_2}$-regular circulant graph pattern where
each processor has $\ceil{\log_2 p}$ incoming and $\ceil{\log_2 p}$
outgoing neighbor processors. Communication is assumed to be only
one-ported (a processor can be engaged in one communication operation
at a time), but to allow a processor to send a block to some processor
and at the same time receive a block from some other
processor~\cite{BarNoyKipnis94,Bruck97}. This {simultaneous
  send/receive} model corresponds to what the combined \mpisendrecv
operation of MPI is meant to accomplish~\cite{MPI-4.1}.

The reduce-scatter (partitioned all-reduce) and allreduce operations have been
intensively researched and many algorithms and implementations, taking
aspects of the communication system (different topologies, different,
hierarchical organizations) into account have been proposed.  A
primary starting point for the algorithms of this paper is the well
known and often used power-of-two, straight doubling, circulant graph,
dissemination allgather (concatenation, all-to-all broadcast)
algorithm by Bruck et al.~\cite{Bruck97}. The algorithms by Bar-Noy
and others for allreduce (computing census
functions)~\cite{BarNoyKipnisSchieber93,BruckHo93} that use a cleverly
adjusted doubling scheme have likewise been a starting point
for subsequent algorithms and generalizations, both for the allreduce
and for the reduce-scatter (partitioned all-reduce)
operation~\cite{BernaschiIannelloLauria03,Traff23:circulant}.

Well-known algorithms assuming either a ring or a fully connected
communication network can solve the problem in $p-1$ communication
rounds, in which each processor sends and receives a partial block
result $W[i]$ to and from a preceding and a succeeding processor and
performs a reduction on the received partial result block, see for
instance~\cite{ChanHeimlichPurkayasthavandeGeijn07,Iannello97,PatarasukYuan09}.
With a ring, the $\oplus$ operator must be commutative whereas with a
fully connected network, the algorithm can also work for
non-commutative operators~\cite{Iannello97}. These algorithms are
optimal in the number of blocks that are sent and received per
processor, namely $p-1$, but have a linear number of communication
rounds which is very unattractive for small block sizes and large
number of processors $p$.  The lower bound on the number of
communication rounds is clearly $\ceil{\log_2 p}$, as is well
known~\cite{Bruck97}.

The reduce-scatter problem can be solved in $\log_2 p$ communication
rounds with the optimal number $p-1$ of sent, received and reduced
partial result blocks with a $\log_2 p$-dimensional hypercube or
butterfly communication pattern. Likewise, the allreduce problem can
be solved with twice as many communication rounds, partial result
blocks and send-receive operations. A drawback of these simple
algorithms is that they do not readily extend to arbitrary numbers of
processors (not only powers-of-two). This problem has often been
addressed and extensions that are better than the trivial reduction to
the nearest power-of-two have been proposed and
implemented~\cite{Traff04:reduction,Traff05:redscat}. Hypercube or
butterfly pattern algorithms with some care work also for
non-commutative operators.

For the reduce (reduction to root) and allreduce operations, pipelined
fixed-degree (binary) trees are also used. Such algorithms are simple
to implement and can work for any number of processors, sometimes also
for non-commutative operators (depending on how trees are constructed
and numbered), but have the disadvantage of losing effective bandwidth
proportional to the arity of the trees. Likewise, there is a latency
penalty proportional to the size of the pipeline block size, which can
in addition be difficult to select well. Some of these problems can be
alleviated by using two trees simultaneously~\cite{Traff09:twotree}.

A standard observation is that allreduce can be accomplished by
performing a reduce-scatter (partitioned all-reduce) operation
followed by an allgather operation. Lower bound arguments
in~\cite{PatarasukYuan09,BarnettLittlefieldPaynevandeGeijn95} show
that when the total number $p(p-1)$ of required applications of the
$\oplus$ operator to blocks of elements are evenly shared by the
processors, it is required to send and receive $2(p-1)$ partial result
blocks per processor. At least $\ceil{\log_2 p}$ communication rounds
are required. Using the reduce-scatter (partitioned all-reduce) and allgather
algorithm of the present paper, the bound on the number of blocks is
achieved with $2\ceil{\log_2 p}$ communication rounds.

We finally observe that all-to-all communication can be accomplished by
a (commutative) reduce-scatter operation by taking concatenation as the
operator.

\section{The Algorithms}

We are given $p$ consecutively ranked processors $r,0\leq r<p$ each of
which can in a communication step simultaneously send a block and
receive a block from two other, possibly different processors. The
reduce-scatter (partitioned all-reduce) and allreduce algorithms
uniformly follow a communication pattern in which each processor $r$
in a communication round $k$ sends a block of elements to a
\emph{to-processor} $(r+s_k)\bmod p$ and receives a block of elements
from a \emph{from-processor} $(r-s_k+p)\bmod p$ for certain skips (or
jumps) $s_k$. A graph $C_p^{s_o,s_1,\ldots s_{q-1}}$ with vertices
$r,0\leq r<p$ and directed edges $(r\pm s_k+p)\bmod p$ is called a
\emph{circulant graph} with \emph{skips} (jumps) $s_k, k=0,\ldots,q-1$
(sometimes a \emph{loop network},
see~\cite{BermondComellasHsu95}). The skips $s_k$ are chosen by
repeated halving of $p$ with rounding up until $s_k=1$,
$s_k=\ceil{s_{k+1}/2}$. The number of such roughly halving steps
required is clearly $q=\ceil{\log_2 p}$.

\subsection{A Simple, Uniform Reduce-scatter Algorithm}

For the reduce-scatter problem, each processor takes an input vector
$V_r$ of elements and each $V_r$ is partitioned in the same way into $p$
disjoint blocks of elements such that each processor has $p$ blocks of
input each with a given, known number of elements. The number of
elements per block may be the same (as in the \mpireducescatterblock
operation) or may be different (as in the \mpireducescatter
operation).  The input blocks for processor $r, 0\leq r<p$ are indexed
as $V_r[i], 0\leq i<p$.  The reduce-scatter operation computes for
each processor $r$ the sum
\begin{eqnarray*}
  W = \bigoplus_{i=0}^{p-1}V_i[r] \quad .
\end{eqnarray*}

A partial result block is any sum of input blocks for some subset of
processors.  A complete reduce-scatter (partitioned all-reduce)
algorithm for processor $r$ is shown as
Algorithm~\ref{alg:blockreduction}. Each processor maintains partial
result blocks $R[i]$ for some $i$ that will contribute towards the
final result $W$ both for processor $r$ itself and for other
processors ranked after $r$ (modulo $p$). The partial result blocks
are maintained in the same way for all processors such that for
processor $r$, $R[i], 0\leq i<p$ is a partial result that will
contribute to the final result at processor $(r+i)\bmod p$. This is
achieved initially by a rotated copy of the input blocks $V[(r+i)\bmod
  p]$ into $R[i]$. In each communication round, a consecutive sequence
of partial result blocks $R[s\ldots s'-1]$ is sent to the to-processor
$(r+s)\bmod p$ and a consecutive sequence with the same number of
blocks is received from the from-processor $(r-s+p)\bmod p$ and added
into $R[0\ldots s'-s-1]$ using the $\oplus$ operator. Thus, for each
$i, 1\leq i<p$, the partial result for the block $R[i]$ is sent once
as part of a consecutive sequence of blocks. Block $R[0]$ is kept as
$W$ and will eventually store the result of the reduction for block
$r$.

\begin{algorithm}
  \caption{The $p$-block reduce-scatter (partitioned all-reduce)
    algorithm for processor $r,0\leq r<p$. Each processor has an input
    vector $V$ of $p$ blocks of elements. Processor $r$ receives in
    $W$ the computed reduction over the $r$th input blocks,
    $W=\oplus_{i=0}^{p-1}V[i]$. The commutative operator for pairwise
    reduction of blocks is $\oplus$.}
  \label{alg:blockreduction}
  \begin{algorithmic}
    \Procedure{PartitionedAllReduce}{$V[p],W$}
    \State $W\gets V[r]$
    \For{$i=1,\ldots,p-1$} $R[i]\gets V[(r+i)\bmod p]$
    \EndFor \Comment $R[0]$ will be kept in $W$
    \State $s\gets p$
    \While{$s>1$}
    \State $s',s\gets s,\ceil{s/2}$ \Comment Halve and round up
    \State $t,f\gets (r+s)\bmod p, (r-s+p)\bmod p$ \Comment To- and from processors
    \State $\bidirec{R[s\ldots s'-1],t}{T[0\ldots s'-s-1],f}$
    \State $W\gets W\oplus T[0]$
    \For{$i=1,\ldots,s'-s-1$}
    $R[i]\gets R[i]\oplus T[i]$
    \EndFor
    \EndWhile
    \EndProcedure
\end{algorithmic}
\end{algorithm}

The sequence of skips (jumps) $s_k$ for the circulant graph are
computed incrementally by halving $s$ from the previous iteration and
rounding up.  It can easily be seen that any $i$ can be written as a
sum of different such skips $s_k\leq i$, which means that for any
processor $r$, there is a(t least one) path from any other processor
$(r-i+p)\bmod p$ consisting of different edges $s_k$. The computation
of partial results are performed along such paths with leaf processors
contributing input block $V[r]$ and interior processors contributing a
partial result including their own input block. For each processor
$r$, there is a spanning tree directed towards $r$ formed by combining
certain such paths along which the result for processor $r$ is
computed. The decomposition of $i$ into sums of different $s_k$ is not
necessarily unique, and depends on $p$. The spanning tree to $r$ is
built incrementally by hooking trees to roots with edges of length $s$
in each iteration.

\begin{theorem}
  \label{thm:blockreduce}
  On $p$ input vectors partitioned into $p$ blocks,
  Algorithm~\ref{alg:blockreduction} solves the reduce-scatter
  (partitioned all-reduce) problem in $\ceil{\log_2 p}$ send-receive
  communication rounds. Each processor sends and receives exactly
  $p-1$ partial result blocks of elements and performs exactly $p-1$
  applications of the commutative reduction operator $\oplus$ on
  partial result blocks.
\end{theorem}

\begin{proof}
  The communication round complexity is obvious, since $q=\ceil{\log_2
    p}$ roughly halving steps are needed for the \textbf{while}-loop
  to terminate. Let $s_k$ be the value of $s$ before the $k$th
  iteration, $k=0,1,\ldots,q-1$.  Starting with $s_0=p$, by the
  repeated halving, clearly $p=s_0>s_1>\ldots>s_{q-1}=1$.  In
  iteration $k$, each processor sends and receives $s_{k}-s_{k+1}$
  blocks of elements and applies the $\oplus$ operator also to this
  number of blocks. Since $\sum_{k=0}^{q-1} (s_k-s_{k+1}) =
  s_0-s_{q-1} = p-1$, the bounds on the communication and computation
  volume follows.

  The algorithm maintains for each processor $r$ the invariant that
  for $0\leq i<s_k$, $R[i]$ (with $W=R[0]$) stores a partial result
  over a subtree $T_i$ rooted at $i$ with subtrees $T_i$ and $T_j$
  being disjoint for $i\neq j$ but spanning all $i,0\leq i<p$. In
  other words, each processor $r$ maintains a spanning forest over all
  $i,0\leq i<p$.  The invariant holds before the first iteration of
  the \textbf{while}-loop since initially each $T_i$ is a singleton
  storing the input $R[i]=V[(r+i)\bmod p]$.  After the last iteration
  where $s_{q-1}=1$ the invariant implies that for each processor $r$,
  $R[0]=W$ is the result of a reduction over a spanning tree $T_0$ of
  the blocks $V_{(r-i+p)\bmod p}[r]$.

  In iteration $k$, subtrees $T_j$ represented by $R[j],s_{k+1}\leq
  j<s_k$ are hooked into subtrees $T_i, 0\leq i<s_k-s_{k+1}$, with
  $j=i+s_{k+1}$, $T_{i+s_{k+1}}\stackrel{s_{k+1}}{\rightarrow} T_i$,
  by an edge labeled with the skip $s_{k+1}$, and the partial sums
  represented by the subtrees $R[i],0\leq i<s_k-s_{k+1}$ updated by
  the partial sums $R[j]$ sent from processor $(r-s_{k+1}+p)\bmod
  p$. Partial results $R[j],1\leq j<p$ are sent once, and $W=R[0]$
  never. Thus, the invariant is maintained, and in each iteration, the
  number of disjoint subtrees decreases by $s'-s$.
\end{proof}

\begin{figure}
  \begin{center}
    \begin{tikzpicture}[scale=0.75, transform shape]
      \node {0}
            [<-]
            child {node {11}}
            child {node {6}
              child {node {17}}
            }
            child {node {3}
              child {node {14}}
              child {node {9}
                child {node {20}}
              }
            }
            child[missing]
            child {node {2}
              child {node {13}}
              child {node {8}
                child {node {19}}
              }
              child {node {5}
                child {node {16}}
              }
            }
            child[missing]
            child[missing]
            child {node {1}
              child {node {12}}
              child {node {7}
                child {node {18}}
              }
              child {node {4}
                child {node {15}}
                child {node {10}
                  child {node {21}}
                }
              }
            }
            ;
    \end{tikzpicture}
  \end{center}
  \caption{The tree implicitly constructed by each processor by
    Algorithm~\ref{alg:blockreduction} for $p=22$.}
  \label{fig:tree22}
\end{figure}

\paragraph{Example:} It is illustrative to trace the way 
Algorithm~\ref{alg:blockreduction} builds the reduction trees for the
blocks. For instance, for $p=22$, the skips are $s=11,6,3,2,1$.  In
each communication round, a tree $T_{i+s}$ is hooked into tree $T_i$
for $0\leq i< s'-s$ by an edge of length $s$. When the algorithm
terminates with $s=1$, a tree as depicted in Figure~\ref{fig:tree22}
has been constructed (implicitly).  For any given processor, say
$r=p-1=21$, the order in which input blocks are reduced can be found
by a depth-first traversal of this tree, subtracting the node label
$i$ from $r$ (modulo $p$) to get the index of the input block
$V[(r-i+p)\bmod p]$.  Processor $r=21$ receives partial results from
processor $21-11=10, 21-6=15, 21-3=18, 21-2=19$ and finally
$21-1=20$. Let $x_i$ denote the input block of processor $i$ for
processor $r$, $x_i=V_i[r]$. Processor $r=21$ then computes
\begin{eqnarray*}
  W = \sum_{i=0}^{p-1} x_i & = & x_{21}+x_{10} + \\
  & & (x_{15}+x_{4}) + \\
  & & (x_{18}+x_{7}+(x_{12}+x_{1})) + \\
  & & (x_{19}+x_{8}+(x_{13}+x_{2})+(x_{16}+x_{5})) + \\
  & & (x_{20}+x_{9}+(x_{14}+x_{3})+(x_{17}+x_{6}+(x_{11}+x_{0}))) \\
\end{eqnarray*}
where each line shows the received partial sums in the five communication
rounds.
\newline

If we assign time costs to the bidirectional send-receive operations,
the total time of the algorithm for $p$ processors and vectors of $m$
elements can be estimated.

\begin{corollary}
  In a homogeneous, linear-affine transmission cost model where
  concurrent, bidirectional sending and receiving blocks of $m/p$
  elements by all processors in a communication round can be charged
  $\alpha+\beta m/p$ time and pairwise reduction of two $m/p$-element
  vectors by the binary $\oplus$ operator takes time $\gamma m/p$, the
  reduce-scatter (partitioned all-reduce) problem on input vectors of $m$
  elements uniformly partitioned into blocks of $m/p$ elements is solved
  in time
  \begin{eqnarray*}
    T(m,p) & = &
    \alpha\ceil{\log_2 p}+\beta\frac{p-1}{p}m+\gamma\frac{p-1}{p}m
  \end{eqnarray*}
  by Algorithm~\ref{alg:blockreduction}. The time for the initial copy
  of the $m$ input elements into $R[i]$ adds another term of at most
  $\gamma m$.
\end{corollary}

The proof of Theorem~\ref{thm:blockreduce} did not use any particular
properties of the roughly halving scheme for computing $s$
except for the fact that it allows any $i,0\leq i<p$ to be written
as a sum of different $s_k$ values of $s$. The algorithm can therefore be
adopted to other communication patterns leading to different number of
communication rounds.

\begin{corollary}
  The reduce-scatter problem can be solved in $q$ communication rounds
  on any circulant graph $C_p^{s_0,s_1,\ldots s_{q-1}}$ with skips
  $s_0>s_1>\ldots >s_{q-1}=1$ provided that any $0<i<p$ can be
  written as a sum of different $s_k, 0<k<q$.
\end{corollary}

Different circulant graphs may be more or less suited to be embedded
into a concrete, given communication network, and some may conceivably
perform better than or different from the roughly halving scheme of
Algorithm~\ref{alg:blockreduction}. It is an open, experimental
question, which sequence of skips may perform best in practice on a
concrete high-performance system.

\paragraph{Examples:} The reduce-scatter problem is solved on a fully
connected network in $p-1$ communication steps by taking
$s_k=p,p-1,p-2,\ldots,1$. This algorithm can easily be made to work
also for non-commutative operators and corresponds to the folklore
algorithm also stated in~\cite{Iannello97}. A straight power-of-two
halving scheme, as used by Bruck et al.~\cite{Bruck97} will lead to
another $\ceil{\log_2 p}$ round algorithm by taking $s_0=p$ and
letting $s_k, k>0$ be the largest power-of-two smaller than
$s_{k-1}$. We can get an algorithm running in a square root of $p$
number of rounds by taking $s_k=p-k\ceil{\sqrt{p}}$ as long as
$s_k>\ceil{\sqrt{p}}$ and for smaller $p$ use a either of the above
schemes.
\newline

Algorithm~\ref{alg:blockreduction} does not make any assumptions on
the way input and result vectors are partitioned into disjoint blocks,
except for requiring that all vectors are partitioned in same
way. Thus, the number of elements in block $V_r[i]$ and block $V_r[j]$
may differ, but the algorithm will work correctly as long as the
number of elements in $V_i[r]$ is equal to the number of elements in
$V_j[r]$ for all $r$. Let $m$ be the total number of elements over all
blocks $V_r[i]$ (which is the same for all processors $r$). Since in
the extreme case, all elements are concentrated in one block only,
partial results of all $m$ elements will be sent and/or received and
reduced in every communication round. This leads to the following
observation.

\begin{corollary}
  On $p$ input vectors of $m$ elements partitioned into $p$ blocks of
  possibly different numbers of elements,
  Algorithm~\ref{alg:blockreduction} solves the reduce-scatter
  (partitioned all-reduce) problem in time at most $\ceil{\log_2
    p}(\alpha+\beta m+\gamma m)$, assuming a homogeneous,
  linear-affine transmission cost model with constant latency $\alpha$
  and transmission and computation cost per unit $\beta,\gamma$,
  respectively.
\end{corollary}

The algorithm can therefore be used also for \mpireducescatter as long
as the sizes of the input blocks do not differ too much, and in the
extreme case of only one block, also for \mpireduce (reduction to
root) as long as the number of elements is not too large compared to
$p$.

For large, irregular reduce-scatter problems where the sizes of the
blocks for the processors can differ significantly, pipelined
algorithms, also using a circulant graph communication pattern, can
perform much better, depending only linearly on the total problem size
$m$, see~\cite{Traff24:logschedule}.

As the example showed, the applications of $\oplus$ are not done in
rank order, and the algorithm assumes and exploits heavily
the commutativity of the $\oplus$ operator. However, all processors
perform the reductions in the same order, which depends on the skips
arising by the repeated halving of $p$. If the input happens to be in
the right order, the algorithm would work for also for a
non-commutative operator. If this is not the case, the input blocks
could be permuted into a suitable order, but this entails a much too
expensive all-to-all redistribution step.

\subsection{An Allreduce Algorithm}

It is easy to see that the allreduce problem can be solved by a
reduce-scatter operation followed by an allgather operation that
gathers together all result blocks at all processors. An allgather
operation can, as classically shown in~\cite{Bruck97}, easily be
implemented by a doubling scheme, essentially the reduce-scatter
algorithm run in reverse.  A variant which exactly reverses the
sequence of skips is shown as Algorithm~\ref{alg:allreduction}.  To
avoid recomputing the skips $s$, they are pushed on a stack during the
reduce-scatter phase and popped in the allgather phase.  This leads
to the following result.

\begin{algorithm}
  \caption{The allreduce algorithm for processor $r,0\leq r<p$ derived
    from the reduce-scatter algorithm by addition of an allgather
    phase that collects all result block on all processors. Each
    processor has an input vector $V$ that can be partitioned into $p$
    blocks of elements.  Each processor $r$ returns in $W$ the
    resulting reduction over all input vectors. The commutative
    operator for pairwise reduction of blocks is $\oplus$.}
  \label{alg:allreduction}
  \begin{algorithmic}
    \Procedure{AllReduce}{$V,W$}
    \State\Comment Assume $V,W$ both partitioned into blocks $V[i],W[i]$
    \For{$i=0,\ldots,p-1$} $R[i]\gets V[(r+i)\bmod p]$
    \EndFor
    \State $s,S \gets p,\bot$ \Comment Empty stack
    \While{$s>1$} \Comment Partitioned all-reduce phase
    \State $\Call{push}{s,S}$ \Comment Push skip $s$ on stack
    \State $s',s\gets s,\ceil{s/2}$ \Comment Halve and round up
    \State $t,f\gets (r+s)\bmod p, (r-s+p)\bmod p$
    \State $\bidirec{R[s\ldots s'-1],t}{T[0\ldots s'-s-1],f}$
    \For{$i=0,\ldots,s'-s-1$}
    $R[i]\gets R[i]\oplus T[i]$
    \EndFor
    \EndWhile
    \While{$S\neq\bot$} \Comment Allgather phase
    \State $s'\gets\Call{pop}{S}$
    \State $f,t\gets (r+s)\bmod p, (r-s+p)\bmod p$
    \State $\bidirec{R[0\ldots s'-s-1],t}{R[s\ldots s'-1],f}$
    \State $s\gets s'$
    \EndWhile
    \For{$i=0,\ldots,p-1$} $W[(r+i)\bmod p]\gets R[i]$
    \EndFor
    \EndProcedure
\end{algorithmic}
\end{algorithm}

\begin{theorem}
  \label{thm:allreduce}
  On $p$ input vectors partitioned into $p$ blocks,
  Algorithm~\ref{alg:allreduction} solves the allreduce problem in
  $2\ceil{\log_2 p}$ send-receive communication rounds. Each processor
  sends and receives exactly $2(p-1)$ blocks of elements and performs
  exactly $p-1$ applications of the commutative reduction operator
  $\oplus$ on blocks of elements.
\end{theorem}

The bound on the number of blocks communicated and the number of
reductions is
optimal~\cite{PatarasukYuan09,BarnettLittlefieldPaynevandeGeijn95}.

\section{Implementation in and for MPI}

Both Algorithm~\ref{alg:blockreduction} and~\ref{alg:allreduction} can
readily be implemented in MPI~\cite{MPI-4.1} with \mpisendrecv or
\mpiisendrecv for the bidirectional, combined $\bidirec{}{}$
operation. Standard considerations as when implementing the doubling
allgather algorithm of Bruck et al.~\cite{Bruck97} apply, see for
instance~\cite{BienzGautamKharel22,Traff06:regallgat}. In particular,
the doubling and halving schemes lead to latency contention and
communication redundancy when run as written on clustered,
hierarchical systems with constrained per node
bandwidth~\cite{Traff20:mpidecomp}.

The algorithms compute the required skips in constant time per
communication rounds. All partial result blocks are kept in
consecutive buffers and no reordering of blocks is needed in the
$\ceil{\log_2 p}$ communication rounds. Reduction and copy
operations can therefore be done as bulk operations over many
blocks. A property of the roughly halving scheme is that no sequence
of blocks is longer than $\ceil{p/2}$. This can be exploited to
avoid half of the copy operations~\cite{Traff23:circulant}. The
standard, straight doubling scheme does not have this
property~\cite{Bruck97}. Explicit copying could be avoided altogether
by using the derived datatype mechanism of MPI~\cite{MPI-4.1}, as done
for all-to-all algorithms in~\cite{Traff14:bruck}; however, copying is
done only before and after the communication rounds, and therefore
this is not likely to be a fruitful implementation choice.
  
\section{Summary}

This paper gave a very simple and easily implementable algorithm
for the reduce-scatter (partitioned all-reduce) operation as found in
message-passing frameworks like MPI~\cite{MPI-4.1}. The algorithm is
optimal in number of communication rounds and in communication and
computation volume. The algorithm was used as building block of an
allreduce algorithm that is likewise optimal in communication and
computation volume.

The circulant graph communication patterns that were used for
reduce-scatter (partitioned all-reduce) and allreduce, can, with some craft,
also be used to solve the all-to-all communication problem in the same
number of communication rounds, similarly to the straight doubling
all-to-all (indexing) algorithms given in~\cite{Bruck97}, namely by
taking the $\oplus$ operator to be concatenation of element blocks.  By
specialization of the algorithms, likewise algorithms for the rooted,
regular scatter and gather problems can easily be derived
(\mpiscatter, \mpigather). We also indicated that the algorithms can
be used for reduction to root (\mpireduce), and by implication, for
broadcast (\mpibcast), and may be attractive for small problem sizes.
Circulant graphs are thus a universal scheme for collective operations
as found in MPI and elsewhere~\cite{Traff23:circulant}.

\bibliographystyle{plainurl}
\bibliography{traff,parallel} 

\end{document}